\documentclass[envcountsame,oribibl]{llncs}

\usepackage{amsmath}
\usepackage{amssymb}
\usepackage{natbib}

\usepackage[dvips]{hyperref} 

\usepackage{algorithmic} %seems incompatible with JACM
\usepackage{algorithm} %must be after hyperref import for alg refs to work right

\newcommand{\oneone}{1\!\!1}
\newcommand{\one}[1]{\oneone\left(#1\right)}

\newcommand{\piecewise}[1]{\left\{\begin{array}{@{\;}ll}#1\end{array}\right.}
\newcommand{\set}[1]{\{#1\}}                        % Set (as in \set{1,2,3})
\newcommand{\setof}[2]{\{\,{#1}\::\:{#2}\,\}}        % Set (as in \setof{x}{x > 0})

\newcommand{\sm}{\setminus} % "set minus"
\newcommand{\N}{\mathbb{N}}                     % Positive integers.
\newcommand{\R}{\mathbb{R}}                     % Reals.
\newcommand{\Z}{\mathbb{Z}}                     % Integers.

\newcommand{\footcomment}[1]{} % uncomment this line to hide author comments
\newcommand{\margincomment}[1]{} % uncomment this line to hide author comments

\newcommand{\bp}{\,\rule[-0.22em]{0.08em}{1em}\,} %binary piecewise, i.e. f \bp g is f on f's domain and g on g's domain

\newcommand{\betTour}{\textsc{BetweennessTour}}
\newcommand{\fast}{\textsc{FAST}}

\begin{document}

\title{Faster Algorithms for Feedback Arc Set Tournament, Kemeny Rank Aggregation and Betweenness Tournament\thanks{A preliminary version of this work appeared in version 1 of the arXiv preprint \cite{Karpinski09betweenness}.}}
\titlerunning{Exact algs}  % abbreviated title (for running head)
\author{Marek Karpinski\inst{1}\thanks{Parts of this work done while visiting Microsoft Research.} \and Warren Schudy\inst{2}\thanks{Parts of this work done while visiting University of Bonn.}}
\authorrunning{Marek Karpinski and Warren Schudy} % abbreviated author list (for running head)
\tocauthor{Marek Karpinski and Warren Schudy}
\institute{University of Bonn,\\
\email{marek@cs.uni-bonn.de}
\and
Brown University,\\
\email{ws@cs.brown.edu}}

\maketitle              % typeset the title of the contribution

%The abstract should summarize the contents of the paper
%using at least 70 and at most 150 words. It will be set in 9-point
%font size and be inset 1.0 cm from the right and left margins.
%There will be two blank lines before and after the Abstract. \dots
\begin{abstract}
We study fixed parameter algorithms for three problems: Kemeny rank aggregation, feedback arc set tournament, and betweenness tournament. For Kemeny rank aggregation we give an algorithm with runtime $O^*(2^{O(\sqrt{OPT})})$, where $n$ is the number of candidates, $OPT \le \binom{n}{2}$ is the cost of the optimal ranking, and $O^*(\cdot)$ hides polynomial factors. This is a dramatic improvement on the previously best known runtime of $O^*(2^{O(OPT)})$. For feedback arc set tournament we give an algorithm with runtime $O^*(2^{O(\sqrt{OPT})})$, an improvement on the previously best known $O^*(OPT^{O(\sqrt{OPT})})$ \cite{Alon09}. For betweenness tournament we give an algorithm with runtime $O^*(2^{O(\sqrt{OPT/n})})$, where $n$ is the number of vertices and $OPT \le \binom{n}{3}$ is the optimal cost. This improves on the previously known $O^*(OPT^{O(OPT^{1/3})}$ \cite{Saurabh09}), especially when $OPT$ is small. Unusually we can solve instances with $OPT$ as large as $n (\log n)^2$ in polynomial time!

\end{abstract}

%%%%%%
\section{Introduction}

Suppose you ran a chess tournament, everybody played everybody (a.k.a.\ round robin) and you wanted to use the results to rank everybody.  Unless you were really lucky, the results would not be acyclic, so you could not just sort the players by who beat whom. A natural objective is to find a ranking that minimizes the number of upsets, where an upset is a pair of players where the player ranked lower in the ranking beat the player ranked higher. Minimizing the number of upsets is called \emph{feedback arc set problem} on tournaments (\fast). The complementary problem of \emph{maximizing} the number of pairs that are \emph{not upsets} is called the \emph{maximum acyclic subgraph problem} on tournaments. These problems are NP-hard~\citep{Ailon08aggregating,Alon06,Charbit07} (see also \citep{Conitzer06a}), but a polynomial-time approximation scheme (PTAS) \cite{Mathieu09fast} is known.

In statistics and psychology, one motivation is \emph{ranking by paired comparisons} \citep{Slater61}: here, you wish to sort some set by some objective but you do not have access to the objective, only a way to compare a pair and see which is greater; for example, determining people's preferences for types of food. This problem attracted computational attention as early as 1961 \citep{Slater61} (for comparison Hoare published quicksort the same year). Feedback arc set tournament and closely related problems have also been used in machine learning~\citep{Cohen99,Ailon08aggregating}.

The \fast{} problem can be generalized to a problem we call \emph{weighted \fast{}}, sometimes known as \emph{feedback arc set with probability constraints}.
The input is a complete directed graph with arc weights $\set{w_{uv}}_{u,v}$ with $w_{uv} + w_{vu} = 1$ for every pair of vertices $u,v$. In other words a weighted \fast{} instance is a convex combination of unweighted \fast{} instances.

We study the \emph{parameterized complexity} of this problem, in particular the parameter $OPT$, the cost of an optimal ranking.  We use the notation $O^*(\cdot)$ to hide factors that are polynomial in the input size, that is $f(I) \in O^*(g(I))$ iff $f(I) \le g(I) |I|^c$ for some $c>0$ and all sufficiently large inputs $I$. The first fixed-parameter algorithms for feedback arc set tournament had runtime $O^*(2^{O(OPT)})$ and later algorithms made a dramatic improvement to $O^*(OPT^{O(\sqrt{OPT})})$ \cite{Alon09}. We improve this to $O^*(2^{O(\sqrt{OPT})})$.

\begin{theorem} \label{thm:exactFAST}
There exists a deterministic parameterized subexponential
   algorithm for weighted \fast{} with runtime $2^{O(\sqrt{OPT})} + n^{O(1)}$. A variant of the algorithm uses $OPT^{O(\sqrt{OPT})} + n^{O(1)}$ time and $n^{O(1)}$ space.
\end{theorem}

The \emph{Exponential time hypothesis} (ETH) \cite{Impagliazzo01} is that 3-SAT cannot be solved in time $2^{o(\text{number of variables})}$. We also give a matching lower bound assuming the ETH:

\begin{theorem} \label{thm:exactLB}
There does not exist a parameterized algorithm for weighted \fast{} with runtime $O^*(2^{o(\sqrt{OPT})})$ unless the exponential time hypothesis \cite{Impagliazzo01} is false.
\end{theorem}

We leave open the possibility that \emph{unweighted} \fast{} may admit an exact algorithm with runtime $O^*(2^{o(\sqrt{OPT})})$.

We note that independently of this work Uri Feige \cite{Feige09} gave an unrelated algorithm for \emph{unweighted} \fast{} matching Theorem \ref{thm:exactFAST}.

\medskip

An important application of weighted feedback arc set tournament is {\em rank aggregation}. Frequently, one has access to several rankings of objects of some sort, such as search engine outputs~\citep{Dwork01}, and desires to aggregate the input rankings into a single output ranking that is similar to all of the input rankings: it should have minimum {average distance} from the input rankings, for some notion of distance. This ancient problem was already studied in the context of voting by \citep{Borda} and \citep{Condorcet} in the 18$^{th}$ century, and has aroused renewed interest recently~\citep{Dwork01,Conitzer06b}.  A natural notion of distance is the number of pairs of vertices that are in different orders, which is known as the Kendall-Tau distance. This defines the {\em Kemeny rank aggregation} problem (KRA) \citep{Kemeny59,Kemeny62}. This choice yields a maximum likelihood estimator for a certain na\"\i{}ve Bayes model~\citep{Young95}. This problem is NP-hard~\citep{Bartholdi89}, even with only four voters \citep{Dwork01}, and has a PTAS \cite{Mathieu09fast}.

We denote the \emph{average} distance between the optimal ranking and the input rankings by $OPT \le \binom{n}{2}$. Two parameters have attacted the bulk of the study: $OPT$ and the average Kendall-Tau distance between the input rankings. It is easy to see (triangle inequality) that these two parameters are within a constant factor of each other, so these parameters give equivalent runtimes up to constants in the exponent. All previous work give algorithms with runtime $O^*(2^{O(OPT)})$ \cite{Betzler09}. There is a standard reduction from Kemeny rank aggregation to weighted \fast{} \cite{Ailon08aggregating,Coppersmith06,Mathieu09fast}, so we improve the best known parameterized algorithm for KRA dramatically to $O^*(2^{O(\sqrt{OPT})})$ as a corollary of our Theorem \ref{thm:exactFAST}.

\begin{corollary} \label{thm:exactKRA}
Let $n$ be the number of candidates and $OPT \le \binom{n}{2}$ the optimum value. There exists a deterministic parameterized subexponential algorithm for Kemeny Rank Aggregation with runtime and space $2^{O(\sqrt{OPT})} + n^{O(1)}$. A variant uses $OPT^{O(\sqrt{OPT})} + n^{O(1)}$ time and $n^{O(1)}$ space.
\end{corollary}

Some other paramters have attracted attention. The parameter of maximum Kendall-Tau distance has been studied but yield bounds no tighter (up to constants in the exponent) than is known for the average Kendall-Tau distance \cite{Betzler09}. Another parameter is the maximum $r_{\max}$, over candidates $c$ and pairs of voters $v_1,v_2$, of the absolute difference between the rank of $c$ in $v_1$ and $v_2$. The best runtime known is $O^*(2^{O(r_{\max})})$ \cite{Betzler09}.

\bigskip

In the Betweenness problem we are given a ground set of \emph{vertices} and a set of \emph{betweenness constraints}
involving $3$ vertices and a \emph{designated} vertex among them.
The objective function of a ranking of the elements is the number of betweenness constraints
for which the designated vertex is not between the other two vertices. The goal is to
\emph{minimize} the objective function. 
For the status of the general Betweenness
problem, see e.g.\ \cite{Opatrny79,Chor98,Ailon07hardness,Charikar09}.
We refer to the Betweenness problem in tournaments, that is in instances with a constraint for every triple of vertices,
as the \betTour{} problem (see \cite{Ailon07hardness}). This problem is NP-hard \cite{Ailon07hardness} and has a recently discovered polynomial-time approximation scheme \cite{Karpinski09betweenness}. We study its parameterized complexity.

\begin{theorem} \label{thm:exactGen}
There exists a randomized parameterized subexponential
   algorithm for \betTour{} with runtime and space $2^{O(\sqrt{OPT/n})} \cdot n^{O(1)}$, where $n$ is the number of vertices and $OPT$ is the cost of the optimal ranking. It succeeds with constant probability.
\end{theorem}

The previously best known runtime was $O^*(2^{O(OPT^{1/3} \log OPT)})$ \cite{Saurabh09}. Our result is better by a logarithmic factor in the exponent for the largest possible $OPT = \Theta(n^3)$ and even better for smaller $OPT$. Interestingly we can solve all instances with $OPT=O(n \log^2 n)$ in polynomial time!

Our results easily generalize to all fully dense ranking CSPs of arity three with \emph{fragile} constraints as introduced by \cite{Karpinski09betweenness}. For simplicity we limit ourselves to the well-known problems discussed above.

\medskip

We now outline the organization of our paper. Section \ref{sec:fast} discusses weighted feedback arc set tournament, including our algorithm (Section \ref{sec:fastAlg}), analysis (\ref{sec:fastAnalysis}), and lower bound (\ref{sec:lb}). Section \ref{sec:bet} discusses our results for betweenness tournament, including our algorithm (Section \ref{sec:betAlg}) and analysis (\ref{sec:betAnalysis}).

\section{Feedback arc set tournament} \label{sec:fast}

\subsection{Algorithm} \label{sec:fastAlg}

We now outline some of our key techniques. Firstly any two low-cost rankings for a FAST problem are nearby in Kendall-Tau distance. Secondly two rankings that are Kendall-Tau distance $D$ apart are equivalent to within additive $O(\sqrt{D})$ in how good each position for each a vertex is (Lemma \ref{lem:fastSVMlandscape}). Thirdly most vertices (in a low-cost instance) have a vee-shaped cost versus position curve and optimal rankings are locally optimal so we know that each vertex belongs at the bottom of its curve. The uncertainty in this curve by $\sqrt{D}$ causes an uncertainty in the optimal position also around $\sqrt{D}$ (Lemmas \ref{lem:feClose} and \ref{lem:psiSmallFAST}). Our algorithm simply computes uncertainties $r(v)$ in the positions of all of the vertices $v$ and solves a dynamic program for the optimal ranking that is near a particular constant-factor approximate ranking. We remark that Braverman and Mossel \cite{Braverman08} and Betzler et al.~\cite{Betzler08,Betzler09} previously applied dynamic programming to \fast{} and KRA.

First we state some core notation. Throughout this paper let $V$ refer to the set of objects (vertices) being ranked and $n$ denote $|V|$. Our $O(\cdot)$ hides absolute constants only. Our $O^*(\cdot)$ hides a polynomial in $n$. A \emph{ranking} is a bijective mapping from a set $S \subseteq V$ to $\set{1,2,3,\ldots,|S|}$. We call $\pi(v)$ the position of $v$ in the ranking $\pi$. We let $d(\pi,\pi')$ denote the Kendall-Tau distance between rankings $\pi$ and $\pi'$, i.e. the number of pairs of vertices in different orders in the two rankings.
An \emph{ordering} is an injection from $S$ into $\R$. We use $\pi$ and $\sigma$ (with superscripts) to denote rankings  and orderings respectively.

The input to weighted \fast{} is a set $V$ of vertices and arc weights $\set{w_{uv}}_{u,v \in V}$ such that $w_{uv} + w_{vu} = 1$ for all $u,v \in V$. 
The \fast{} objective function is the weight of the backwards arcs $C(\pi)=\sum_{u,v \in V : \pi(v)>\pi(u)} w_{vu}$. For ranking $\pi$, vertex $v \in V$ and $p \in \R$ (with $\pi(u) \ne p$ for all $u \ne v$) we define $b(\pi,v,p) = \sum_{u \ne v} \piecewise{w_{vu} & \text{if } p > \pi(u) \\ w_{uv} &  \text{if } p < \pi(u)}$, i.e.\ the cost of the arcs incident to $v$ in the ordering formed by moving $v$ to position $p$ in $\pi$.
Let $\pi^*$ denote an optimal ranking and $OPT=C(\pi^*)$ its cost.

%%%%%%%%%%%%%%%%%%%%%%%%%%%%%%
\begin{algorithm}[t]
Input: Vertex set $V$, arc weights $\set{w_{uv}}_{u,v\in V}$.
\begin{algorithmic}[1]
\STATE Sort by weighted indegree \cite{Coppersmith06}, yielding ranking $\pi^1$ of $V$.
\STATE Set $r(v) = 4\sqrt{2C(\pi^1)} + 2b(\pi^1,v,\pi^1(v))$ for all $v \in V$.
\STATE Use dynamic programming or divide-and-conquer (Details: Lemma \ref{lem:DP}) to find the optimal ranking $\pi^2$ with $|\pi^2(v) - \pi^1(v)| \le r(v)$ for all $v$.
\end{algorithmic}
\caption{Exact algorithm for FAST. If dynamic programming is used in the last line the runtime and space are both $n^{O(1)} 2^{O(\sqrt{OPT})}$. If divide-and-conquer is used the runtime is $n^{O(\sqrt{OPT})}$ and the space is $n^{O(1)}$.}
\label{alg:exactFAST}
\end{algorithm}
%%%%%%%%%%%%%%%%%%%%%%%%%%%%

\medskip

Before running our main Algorithm \ref{alg:exactFAST} we compute a small \emph{kernel}, that is a smaller instance with the same optimal cost as the input instance (up to a known shift). This preliminary step allows us to separate the dependence on $n$ and $OPT$ in the runtime, yielding the runtime stated in Theorem \ref{thm:exactFAST}.

Dom et al.\ \cite{Dom06} give an algorithm for computing kernels of \emph{unweighted} FAST instances with $O(OPT^2)$ vertices. This was later improved to $O(OPT)$ vertices by Bessy et al. \cite{Bessy09}. There is a kernelization algorithm for Kemeny rank aggregation in \cite{Betzler09}, but it produces an instance of size $O(\text{(Number of voters)} \cdot OPT)$, not the desired $OPT^{O(1)}$. To get the desired kernel for general weighted \fast{} we consider a slight variant of the algorithm from \cite{Dom06}.

\begin{lemma}\label{lem:kernel}
There is polynomial-time computable $O(OPT^2)$-vertex kernel for weighted FAST. 
\end{lemma}
\begin{proof}[sketch]
Let $OPT \le U \le 5OPT$ be the cost of a 5-approximate ranking \cite{Coppersmith06}.

We say that an arc is a \emph{majority arc} if it has greater weight than its reverse, with ties broken arbitrarily. A majority arc clearly has weight at least 1/2. The \emph{majority tournament} \cite{Ailon08aggregating} is the unweighted directed graph with vertex set $V$ and arc set equal to the majority arcs.

Our kernelization algorithm is simple: we apply the following two reduction rules, which are extensions of two reduction rules in \cite{Dom06}, as often as possible.

The first reduction rule is eliminating a vertex that is part of no cycles of three arcs in the majority tournament. Consider some such vertex $v$. It is easy to see that there exists an optimal ranking that puts every predecessor of $v$ (in the majority tournament) before $v$ and every successor of $v$ after $v$, while implies the validity of this rule.

The second reduction rule concerns an arc $(u,v)$ of the majority tournament that is in more than $2U$ cycles of three arcs in the majority graph. Any feedback arc set not not paying for such an arc must pay for more than $2OPT$ other arcs of the majority tournament, each of cost at least 1/2, and hence cannot be optimal. Therefore we record that we must pay $w_{uv}$, then set weight $w_{uv}$ to zero and $w_{vu}$ to one.

Now we argue that the resulting instance after these two rules are exhaustively applied has $O(OPT^2)$ vertices.
An optimal feedback arc set, which necessarily has cost $OPT$, can include at most $2OPT$ majority arcs. 
Each such majority arc is in at most $10OPT$ triangles by the second rule, so there are at most $20OPT^2$ triangles.
Finally by the first rule every vertex is in a triangle, so there are at most $60OPT^2$ vertices.
\end{proof}

We have not investigated whether or not the $O(OPT)$ vertex kernel for unweighted FAST \cite{Bessy09} can be extended to weighted FAST.

%%%%%%%%%%%%%%%%%%%%%%
\subsection{Analysis}  \label{sec:fastAnalysis}

Variants of the following Lemma are given in \cite{Mathieu09fast} and \cite{Karpinski09betweenness}. We give a simplified proof here for completeness.

\begin{lemma}[\cite{Mathieu09fast,Karpinski09betweenness}]\label{lem:fastSVMlandscape}
Let $\pi$ and $\pi'$ be rankings over $V$. It follows that $|b(\pi, v, p) - b(\pi', v, p)| \le 2 \sqrt{d(\pi,\pi')}$ for all $v \in V$ and $p \in \R \sm \Z$.
\end{lemma}

\begin{proof}
Fix $v \in V$, $p \in \R \sm \Z$ and rankings $\pi$,$\pi'$. 
Consider the sets of vertices $L=\setof{u \in V \sm \set{v}}{\pi(u)<p<\pi'(u)}$ and $R=\setof{u \in V \sm \set{v}}{\pi'(u)<p<\pi(u)}$. Intuitively these are the vertices that cross $p$ from left to right (resp. right to left) when going from $\pi$ to $\pi'$. It follows easily from the definition of $b$ that $|b(\pi, v, p) - b(\pi', v, p)| \le |L|+|R|$, so we now proceed to bound $|L|$ and $|R|$.

The bijective nature of $\pi$ and $\pi'$ implies that $|L|=|R|$. Observe that all vertices in $L$ are before all vertices in $R$ in $\pi$, and vice versa for $\pi'$, hence $d(\pi,\pi') \ge |L||R|$. Putting these facts together proves the Lemma.
\end{proof}

\begin{lemma}\label{lem:feClose}
In Algorithm~\ref{alg:exactFAST} we have  $|\pi^*(v) - \pi^1(v)| \le r(v)$ for all $v \in V$ and any optimal ranking $\pi^*$ of $V$.
\end{lemma}
\begin{proof}
The weight of an arc and its reverse sum to one so $d(\pi^*, \pi^1) \le C(\pi^*) + C(\pi^1) \le 2C(\pi^1)$. By Lemma \ref{lem:fastSVMlandscape} therefore 
\begin{equation}
|b(\pi^*, v, j + 1/2) - b(\pi^1, v, j +1/2)| \le 2\sqrt{2C(\pi^1)} \label{eqn:feClose}
\end{equation}
for any $j \in \Z$.

Fix $v \in V$. We conclude
\begin{align*}
\lefteqn{|\pi^*(v) - \pi^1(v)|} \\
&\le b(\pi^1, v, \pi^*(v)) + b(\pi^1, v, \pi^1(v)) && \text{($w_{uv}+w_{vu}=1$ for all $u$)}\\
 &=  b(\pi^1, v, \pi^*(v) + 1/2) + b(\pi^1, v, \pi^1(v) + 1/2) && \text{($\pi^1$ is integral)}\\
 &\le b(\pi^*, v, \pi^*(v) + 1/2) + 2 \sqrt{2C(\pi^1)} + b(\pi^1, v, \pi^1(v) + 1/2)&& \text{(By (\ref{eqn:feClose}))} \\
 &\le b(\pi^*, v, \pi^1(v) + 1/2) + 2 \sqrt{2C(\pi^1)} + b(\pi^1, v, \pi^1(v) + 1/2) && \text{(Optimality of }\pi^*\text{)} \\
&\le 4 \sqrt{2C(\pi^1)} + 2b(\pi^1, v, \pi^1(v) + 1/2) && \text{(By (\ref{eqn:feClose}))} \\
& = r(v) && \text{(Definition of $r(v)$)}
.\end{align*}
\end{proof}

\begin{lemma}\label{lem:psiSmallFAST}
In Algorithm~\ref{alg:exactFAST} we have $\max_{j \in \Z} |\setof{v \in V}{|\pi^1(v) - j| \le r(v)}| = O(\sqrt{OPT})$.
\end{lemma}
\begin{proof}
Fix $j \in \Z$. Let $R = \setof{v \in V}{|\pi^1(v) - j| \le r(v)}$, the cardinality of which we are trying to bound. We say $v \in V$ is \emph{pricey} if $2b(\pi^1,v,\pi^1(v)) > \sqrt{2C(\pi^1)}$. Clearly $2 C(\pi^1) = \sum_v b(\pi^1, v, \pi^1(v)) \ge (\text{number pricey}) \frac{1}{2}\sqrt{2C(\pi^1)}$ hence the number of pricey vertices is at most $\frac{2 C(\pi^1)}{(1/2)\sqrt{2C(\pi^1)}}=2\sqrt{2C(\pi^1)}$. All non-pricey vertices in $R$ have $|\pi^1(v) - j| \le r(v) \le 5 \sqrt{2C(\pi^1)}$, so at most $10 \sqrt{2C(\pi^1)} + 1$ non-pricey vertices are in $R$. We conclude $|R| \le 12 \sqrt{2C(\pi^1)} + 1 = O(\sqrt{OPT})$ since $\pi^1$ is a 5-approximation \cite{Coppersmith06}.
\end{proof}

\begin{lemma}\label{lem:DP}
There is a dynamic program for \fast{} that finds the optimal ranking $\pi^2$ with $|\pi^2(v) - \pi^1(v)| \le r(v)$ for all $v$ using space and runtime $O(|V|^2) 2^{\psi}$, where $\psi = \max_j |\setof{v \in V}{|\pi^1(v) - j| \le r(v)}|$. A divide and conquer variant uses $|V|^{O(\psi)}$ time and $|V|^{O(1)}$ space.
\end{lemma}

\begin{proof}
Say that a set $S \subseteq V$ is \emph{valid} if it contains all vertices $v$ with $\pi^1(v) \le |S| - r(v)$ and no vertex $v$ with $\pi^1(v) > |S| + r(v)$. Observe that for any $s \in \N$ all valid sets of size $s$ agree except for the presence or absence of $\psi$ vertices.  Therefore there are at most $n 2^\psi$ valid sets.

We say that a ranking $\pi$ of valid set $S$ is \emph{valid} if $\setof{v}{\pi(v) \le j}$ is a valid set for all $0 \le j \le |S|$. It is easy to see that a ranking $\pi$ is valid if and only if satisfies $|\pi(v) - \pi^1(v)| \le r(v)$ for all $v$.

One can easily see the following optimal substructure property: prefixes of an optimal  valid ranking are optimal valid rankings themselves.

For any valid set $S$ let $\bar C(S)$ denote the cost of the optimal valid ranking of $S$. The recurrence relation is
\[
\bar C(S) = \min_{v \in S : S \sm \set{v} \text{ is valid}} \left[ \bar C(S \sm \set{v}) + \sum_{u \in S \sm \set{v}} w_{vu} \right]
.\]

The space-efficient variant evaluates $\bar C$ using divide and conquer instead of dynamic programming, similar to \cite{Dom06}. Details deferred.
\end{proof}

Now we put the pieces together and prove Theorem~\ref{thm:exactFAST}.
\begin{proof}[of Theorem~\ref{thm:exactFAST}]
The kernelization algorithm of Lemma \ref{lem:kernel} allows us to assume without loss of generality that $n=O(OPT^2)$. Algorithm \ref{alg:exactFAST} returns an optimal ranking by Lemmas  \ref{lem:feClose} and \ref{lem:DP}. Lemmas \ref{lem:psiSmallFAST} and \ref{lem:DP} allow us to bound the runtime and space requirements of the dynamic program.
\end{proof}

%%%%%%%%%%%%%%%%%%%%
\subsection{Lower bound}  \label{sec:lb}

%The sparsification lemma \cite{Impagliazzo01} shows that there is a $2^{o(\text{number of clauses})}$-time algorithm for 3-SAT iff there is a $2^{o(\text{number of variables})}$-time algorithm.

\begin{proof}[of Theorem \ref{thm:exactLB}]
For sake of contradiction suppose we have an algorithm for weighted FAST with runtime $2^{o(\sqrt{OPT})}$. We present a series of reductions which converts such an algorithm into a subexpontial-time algorithm for vertex cover, the existence of which is known to contradict the ETH \cite{Flum06}.

Let an instance of vertex cover with $n$ vertices be given. Applying Karp's reduction from vertex cover to feedback arc set \cite{Karp72} produces a feedback arc set instance with $2n$ vertices. Finally one can reduce this to a weighted FAST instance with the same number of vertices by representing incomparable pairs of vertices by opposite arcs of weight 1/2. The result is an weighted FAST instance with $2n$ vertices that is equivalent to the original vertex cover instance. The optimal cost for this instance is at most its number of arcs, which is $O(n^2)$, so the hypothesized algorithm has runtime $2^{o(\sqrt{OPT})}=2^{o(n)}$. This runtime is subexponential, contradicting the ETH.
\end{proof}

%%%%%%%%%%%%%%%%%%%%%%%%%%%%
\section{Betweenness tournament}  \label{sec:bet}

\subsection{Algorithm}  \label{sec:betAlg}

We now introduce some new notation for the betweenness problem. We let $\binom{n}{k}$ (for example) denote the standard binomial coefficient and $\binom{V}{k}$ denote the set of subsets of set $V$ of size $k$. For any ordering $\sigma$ let $\text{Ranking}(\sigma)$ denote the ranking naturally associated with $\sigma$.

Let $v \mapsto p$ denote the ordering over $\set{v}$ which maps $v$ to $p$.
For set $Q$ of vertices and ordering $\sigma$ with domain including $Q$ let $Q \mapsto \sigma$ denote the ordering over $Q$ which maps $u \in Q$ to $\sigma(u)$, i.e.\ the restriction of $\sigma$ to $Q$. For orderings $\sigma^1$ and $\sigma^2$ with disjoint domains let  $\sigma^1\bp \sigma^2$ denote the natural combined ordering over $Domain(\sigma^1) \cup Domain(\sigma^2)$. For example of our notations, $Q \mapsto \sigma \bp v \mapsto p$ denotes the ordering over $Q \cup \set{v}$ that maps $v$ to $p$ and $u \in Q$ to $\sigma(u)$.

A ranking 3-CSP consists of a ground set $V$ of \emph{vertices} and a \emph{constraint system} $c$, where $c$ is a function from rankings of 3 vertices to $[0,1]$. For brevity we henceforth abuse notation and and write $c(\text{Ranking}(\sigma))$ by $c(\sigma)$. The objective of a ranking CSP is to find an ordering $\sigma$ (w.l.o.g.\ a ranking) minimizing $C(\sigma) = \sum_{S \in \binom{\text{Domain}(\sigma)}{k}} c(S \mapsto \sigma)$. We will only ever deal with one constraint system $c$ at a time, so we leave the dependence of $C$ on $c$ implicit in our notations. Abusing notation we sometimes refer to $S \subseteq V$ as a \emph{constraint}, when we really are referring to $c(S \mapsto \cdot)$.
Clearly one can model \betTour{} as a ranking 3-CSP.

Let $b(\sigma, v, p) = \sum_{Q : \cdots} c(Q \mapsto \sigma \bp v \mapsto p)$, where the sum is over sets $Q \subseteq Domain(\sigma) \sm \set{v}$ of size 2. Note that this definition is valid regardless of whether or not $v$ is in $\text{Domain}(\sigma)$. The only requirement is that the range of $\sigma$ excluding $\sigma(v)$ must not contain $p$. This ensures that the argument to $c(\cdot)$ is an ordering (injective). 

Our algorithm and analysis for \betTour{} are analogous to our results for \fast{} with two major differences. Firstly no kernel for betweenness tournament is known, which hurts the runtime somewhat. Secondly we use a more complicated approach to get the preliminary constant-factor approximation ranking $\pi^1$.
We use the known PTAS \cite{Karpinski09betweenness} with an appropriate error parameter to get a 2-approximation. Our analysis requires not only that $\pi^1$ be of \emph{cost} comparable to $\pi^*$ but also that it be close in \emph{Kendall-Tau distance}. Fortunately the analysis of the PTAS from \cite{Karpinski09betweenness} supports the following theorem.

\begin{theorem}[\cite{Karpinski09betweenness}] \label{thm:nearby}
There exists a polynomial-time algorithm for \betTour{} that produces a set $\Pi$ of $O(1)$ rankings. With constant probability one of the rankings $\pi\in\Pi$ satisfies  $d(\pi,\pi^*)=O(OPT/n)$ and has cost at most $2C(\pi^*)$, where $\pi^*$ is some optimal ranking.
\end{theorem}

%%%%%%%%%%%%%%%%%%%%%%%%%%%%%%
\begin{algorithm}[t]
Input: Vertex set $V$
\begin{algorithmic}[1]
 \STATE Use the Algorithm from Theorem \ref{thm:nearby} to construct a set of rankings $\Pi$
 \STATE Let $\pi^{good}$ be the ranking from $\Pi$ with lowest cost
 \FOR{each $\pi^1 \in \Pi$}
 \IF{$C(\pi^1) \le 2 C(\pi^{good})$}
  \STATE Set $r(v) = \alpha_1\sqrt{C(\pi^1) / n} + \alpha_2 b(\pi^1,v,\pi^1(v))/n$ for all $v \in V$, where $\alpha_1$ and $\alpha_2$ are absolute constants.
  \STATE Use dynamic programming (see Lemma \ref{lem:DP}) to find the optimal ranking $\pi^2$ with $|\pi^2(v) - \pi^1(v)| \le r(v)$ for all $v$.
 \ENDIF
\ENDFOR
\STATE Return the best of the $\pi^2$ rankings.
\end{algorithmic}

\caption{Our algorithm for \betTour. The runtime is $n^{O(1)} 2^{O(\sqrt{OPT/n})}$. }
\label{alg:exactGen}
\end{algorithm}
%%%%%%%%%%%%%%%%%%%%%%%%%%%%

%%%%%
\subsection{Analysis}  \label{sec:betAnalysis}

The following two lemmas are given in \cite{Karpinski09betweenness} in more generality. We give simplified proofs here for readability.

\begin{lemma}[\cite{Karpinski09betweenness}]\label{lem:bChangeFirst}
For any rankings $\pi$ and $\pi'$ over vertex set $V$, vertex $v \in V$ and $p \in \R$ we have
\begin{align*}
|b(\pi, v, p) - b(\pi', v, p)|
&\le 3(n-1)\sqrt{d(\pi, \pi')}
.\end{align*}
\end{lemma}

\begin{proof}
Fix $\pi$, $\pi'$, $v$,  and $p$. 
As in the proof of Lemma \ref{lem:fastSVMlandscape}, consider the sets of vertices $L=\setof{u \in V \sm \set{v}}{\pi(u)<p<\pi'(u)}$ and $R=\setof{u \in V \sm \set{v}}{\pi'(u)<p<\pi(u)}$.
From the definition of $b$ we see that a constraint $\set{u,u',v}$ contributes identically to $b(\pi, v, p)$ and $b(\pi', v, p)$ unless either:
\begin{enumerate}
\item $\set{u,u'}$ and $(L \cup R)$ have a non-empty intersection (or)
\item $\one{\pi(u) < \pi(u')} \ne \one{\pi'(u) < \pi'(u')}$.
\end{enumerate}

In the proof of Lemma \ref{lem:fastSVMlandscape} we showed that $|L|=|R|\le \sqrt{d(\pi,\pi')}$.
We can therefore bound
\begin{equation}
|b(\pi, v, p) - b(\pi', v, p)| \le (2 \sqrt{d(\pi,\pi')})(n-2) + d(\pi,\pi') \label{eq:hello}
.\end{equation}
Clearly $d(\pi,\pi') \le n(n-1)/2$, hence $d(\pi,\pi') =(\sqrt{d(\pi,\pi')})^2 \le \sqrt{d(\pi,\pi')} \sqrt{\frac{n(n-1)}{2}} \le (n-2)\sqrt{d(\pi,\pi')}$ for sufficiently large $n$.  Substituting this inequality into the second term of (\ref{eq:hello}) proves the Lemma.
\end{proof}

\begin{lemma}[\cite{Karpinski09betweenness}]\label{lem:fragile} Let $\pi$ be a ranking of $V$, $|V|=n$, $v\in V$ be a vertex and $p,p' \in \R$.
Let $B$ be the set of vertices (excluding $v$) between $p$ and $p'$ in $\pi$. 
Then
$b(\pi, v, p) + b(\pi, v, p') \ge \frac{(n-2)|B|}{2}$. 
\end{lemma}

\begin{proof}
By definition
\begin{equation}
b(\pi, v, p) + b(\pi, v, p') = \sum_{Q : \cdots} \left[ c(Q \mapsto \pi \bp v \mapsto p)+c(Q \mapsto \pi \bp v \mapsto p') \right] \label{eqn:pen2}
\end{equation}
where the sum is over sets $Q \subseteq V \sm \set{v}$ of $2$ vertices. Observe that betweenness tournament has a special property: the quantity in brackets in (\ref{eqn:pen2}) is at least 1 for every $Q$ that has at least one vertex between $p$ and $p'$ in $\pi$. There are at least $|B| (n-2)/2$ such sets.
\end{proof}

\begin{lemma}\label{lem:feCloseGen}
During the iteration of Algorithm~\ref{alg:exactGen} that considers the ranking with $d(\pi^1,\pi^*)=O(OPT/n)$ and $C(\pi^1) \le 2C(\pi^*)$ guaranteed by Theorem \ref{thm:nearby} we have  $|\pi^*(v) - \pi^1(v)| \le r(v)$ for all $v \in V$.
\end{lemma}
\begin{proof}
By Lemma \ref{lem:bChangeFirst} and Theorem \ref{thm:nearby} we have 
\begin{equation}
|b(\pi^*, v, j + 1/2) - b(\pi^1, v, j +1/2)| = O( n\sqrt{OPT / n}) \label{eqn:feCloseGen}
\end{equation}
for any $j \in \Z$.

Fix $v \in V$. We conclude
\begin{align*}
\lefteqn{|\pi^*(v) - \pi^1(v)|\frac{n-2}{2}} \\
&\le b(\pi^1, v, \pi^1(v) + 1/2) + b(\pi^1, v, \pi^*(v) + 1/2) && \text{(Lemma \ref{lem:fragile})}\\
 &\le b(\pi^*, v, \pi^*(v) + 1/2) + O(\sqrt{nOPT}) + b(\pi^1, v, \pi^1(v) + 1/2)&& \text{(By (\ref{eqn:feCloseGen}))} \\
 &\le b(\pi^*, v, \pi^1(v) + 1/2) + O(\sqrt{nOPT}) + b(\pi^1, v, \pi^1(v) + 1/2) && \text{(Optimality of }\pi^*\text{)} \\
&\le O( \sqrt{nOPT}) + 2b(\pi^1, v, \pi^1(v) + 1/2) && \text{(By (\ref{eqn:feCloseGen}))} \\
& = r(v)\frac{n-2}{2} && \text{(Definition of $r(v)$)}
.\end{align*}
\end{proof}

\begin{lemma}\label{lem:psiSmallGen}
In Algorithm \ref{alg:exactGen} we have  $\max_{j \in \Z} |\setof{v \in V}{|\pi^1(v) - j| \le r(v)}| = O(\sqrt{C(\pi^1) / n})$.
\end{lemma}
\begin{proof}
We proceed analogously to the proof of Lemma \ref{lem:psiSmallFAST}.
Fix $j$. Let $R = \setof{v \in V}{|\pi^1(v) - j| \le r(v)}$, whose cardinality we are trying to bound. We say $v \in V$ is \emph{pricey} if $b(\pi^1,v,\pi^1(v))/n > \sqrt{2C(\pi^1)/n}$. Clearly $3 C(\pi^1) = \sum_v b(\pi^1, v, \pi^1(v)) \ge (\text{number pricey}) n \sqrt{2C(\pi^1)/n}$ hence the number of pricey vertices is at most $3C(\pi^1)/(\sqrt{2 n C(\pi^1)})=O(\sqrt{C(\pi^1)/n})$. All non-pricey vertices in $R$ have $|\pi^1(v) - j| = O(\sqrt{C(\pi^1) / n})$, so $O(\sqrt{C(\pi^1) / n})$ non-pricey vertices are in $R$. We conclude $|R| = O(\sqrt{C(\pi^1) / n})$.
\end{proof}

\begin{lemma}\label{lem:DPGen}
There is a dynamic program for betweenness that finds the optimal ranking $\pi^2$ with $|\pi^2(v) - \pi^1(v)| \le r(v)$ for all $v$, with space and runtime $O(|V|^3 2^{\psi})$ where $\psi = \max_j |\setof{v \in V}{|\pi^1(v) - j| \le r(v)}|$. A divide and conquer variant uses $|V|^{O(\psi)}$ time and $|V|^{O(1)}$ space.
\end{lemma}

\begin{proof}
As in the proof of Lemma \ref{lem:DP} we say that a set $S \subseteq V$ is \emph{valid} if it contains all vertices $v$ with $\pi^1(v) \le |S| - r(v)$ and no vertex $v$ with $\pi^1(v) > |S| + r(v)$. Observe that for any $s \in \Z$ the valid sets of size $s$ are identical except for the presence or absence of at most $\psi$ vertices. Therefore there are at most $n 2^\psi$ valid sets. 

We say that a ranking $\pi$ of valid set $S$ is \emph{valid} if $\setof{v}{\pi(v) \le j}$ is a valid set for all $0 \le j \le |S|$. It is easy to see that a ranking $\pi$ is valid if and only if satisfies $|\pi(v) - \pi^1(v)| \le r(v)$ for all $v$.

The dynamic program based on $C(\cdot)$ that worked for \fast{} does not appear to directly generalize to \betTour. We consider an alternate approach.
For any ranking $\pi$ over $S$ let $C'(\pi)$ denote the portion of the cost shared by all orderings with prefix $\pi$. That is, the cost of all constraints with at most 1 vertex outside $S$. One can easily see the following optimal substructure property: prefixes of an optimal (w.r.t. $C'$) valid ranking are optimal (w.r.t. $C'$) valid rankings themselves.

For any valid set $S$ let $\kappa(S)$ denote the $C'$ cost of the optimal (w.r.t. $C'$) valid ranking of $S$. The recurrence relation is
\[
\kappa(S) = \min_{v \in S : S \sm \set{v} \text{ is valid}} \left[ C'(S \sm \set{v}) + \sum_{u \in S \sm \set{v}}\sum_{q \in V \sm S} c(u \mapsto 1 \bp v \mapsto 2 \bp q \mapsto 3) \right]
.\]
\end{proof}

\begin{proof}[of Theorem~\ref{thm:exactGen}]
Lemmas \ref{lem:DPGen} and \ref{lem:psiSmallGen}, plus the test of the ``if'' in Algorithm  \ref{alg:exactGen}, allow us to bound the runtime and space requirements of the dynamic program used by Algorithm \ref{alg:exactGen} by $n^{O(1)} 2^{O(\sqrt{C(\pi^{good})/n})}$, which is of the correct order since $C(\pi^{good}) \le 2 C(\pi^*)$. The ``for'' loop is over a constant number of options and hence does not impact the runtime.

For correctness we focus on the iteration of Algorithm~\ref{alg:exactGen} that considers the $\pi^1 \in \Pi$ with $d(\pi^1, \pi^*) = O(\sqrt{C(\pi^*) / n})$ and $C(\pi^1) \le 2 C(\pi^*)$ as guaranteed by Theorem \ref{thm:nearby}. Theorem \ref{thm:nearby} ensures $C(\pi^1) \le 2 C(\pi^*) \le 2 C(\pi^{good})$ and hence the ``if'' is passed. By Lemma \ref{lem:feCloseGen} $\pi^*$ is among the orders the dynamic program considers.
\end{proof}

%%%%%%%%%%%
\section*{Acknowledgements}
We would like to thank Venkat Guruswami, Claire Mathieu, Prasad Raghavendra and Alex Samorodnitsky for interesting remarks and discussions.

%%%%%%%%%%%
%\nocite{KS09,CS98,KS09,AA07,ALS09,RV07,Mathieu08,FK99,FKK06,BFK03,AKK95,AFKK02,CGM09}   %shows refs
\nocite{Bredereck09}
%%%%%%%%%%%%%%%%%%%%%%%%%%%%%%%%%%%% BIB %%%%%%%%%%%%%%%%%%%%%%%%%%%%%%%%%%
\bibliographystyle{abbrvnat}
\bibliography{big_bib}  %name of bib file

%\appendix

%%%%%%%%%%%%%%
\end{document}